\documentclass[runningheads, a4paper]{llncs}
\usepackage[T1]{fontenc}
\usepackage{graphicx}
\usepackage{algpseudocode}
\usepackage{algorithm}
\usepackage{amsmath}
\usepackage{amssymb}

\usepackage{xspace}
\usepackage{hyperref}
\usepackage{graphicx}
\usepackage[table]{xcolor}
\usepackage{arydshln}
\usepackage{anyfontsize}
\usepackage{pifont}
\usepackage{subfig}

\usepackage[basic]{complexity}

\usepackage[capitalize]{cleveref}
\usepackage{tikz}

\usetikzlibrary{positioning}
\usetikzlibrary{arrows}
\usetikzlibrary{calc}
\usetikzlibrary{shapes}
\usetikzlibrary{automata}

\usetikzlibrary{decorations.pathreplacing}
\usetikzlibrary{decorations.text}

\usetikzlibrary{fit}
\usetikzlibrary{backgrounds}
\usetikzlibrary{calc}
\usetikzlibrary{arrows.meta}

\usetikzlibrary{trees}
\usetikzlibrary{shapes}
\usetikzlibrary{decorations.pathmorphing}
\usetikzlibrary{decorations.markings}
\usetikzlibrary{arrows}

\usetikzlibrary{arrows.meta}
\usetikzlibrary{shapes.geometric}
\usetikzlibrary{fit}
\tikzset{>=stealth, shorten >=1pt}
\tikzset{every edge/.style = {thick, ->, draw}}
\tikzset{every loop/.style = {thick, ->, draw}}

\tikzset{every state/.style={
            very thick,
            fill=black!10,
            rounded corners=1mm,
            minimum size=6mm,inner sep=2pt,
}}

\tikzset{anode/.style={
        font=\small,
        rectangle
}}
\tikzset{ostate/.style={state,rectangle}}
\tikzset{estate/.style={state,
            minimum size=8mm,
diamond}}

\tikzset{every picture/.style={
            >={Triangle},
            thick,
            ->,
            initial text={},
            initial distance=2.5em,
            thick
}}

\tikzset{every edge/.append style={
  shorten <=2pt,
  shorten >=2pt,
}}
\tikzset{every loop/.append style={
  shorten <=2pt,
  shorten >=2pt,
}}

\tikzset{every label/.style={
        text=black!50,
        font=\small,
}}

\tikzset{
         strike through/.append style={
    decoration={markings, mark=at position 0.5 with {
    \draw[-] ++ (-5pt,-5pt) -- (5pt,5pt);}
  },postaction={decorate}}
}

\newcommand{\N}{\mathbb{N}}

\renewcommand{\epsilon}{\varepsilon}

\newcommand{\x}{\times}

\newcommand{\length}[1]{\lvert #1 \rvert}

\newcommand{\card}[1]{\lvert #1 \rvert}

\newcommand{\eqby}[2][=]{\stackrel{\text{{\tiny{#2}}}}{#1}}
\newcommand{\eqdef}{\eqby{def}}

\newcommand{\statespart}{V\mkern1.5mu{=}\mkern1.5mu V_1\mkern1.0mu{\uplus}\mkern1.0mu V_2}

\newcommand{\strategy}{\sigma}

\newcommand{\step}[2][]{\Step{#2}{}{#1}}
\newcommand{\Step}[3]{\ensuremath{\,{\stackrel{#1}{\longrightarrow}}{}^{\scriptstyle{#2}}_{\scriptstyle{#3}}}\,}

\newcommand{\maxtime}{K}
\newcommand{\states}{V}
\newcommand{\nodes}{V}
\newcommand{\edges}{E}

\newcommand{\colours}{C}
\newcommand{\col}{\mathop{col}}
\newcommand{\true}{\mathop{true}}

\newcommand{\hor}{h}
\newcommand{\period}{K}
\newcommand{\init}{s_0}

\newcommand{\ReachCols}[2]{R_{#2}^{#1}}

\newcommand{\PONE}{Player~1\xspace}
\newcommand{\PTWO}{Player~2\xspace}

\newcommand{\Pre}[1][1]{\mathop{Pre}\nolimits_{#1}}
\newcommand{\Post}[1]{\mathop{Post}(#1)}

\begin{document}
\title{Parity Games on Temporal Graphs}
%
%
\author{
    Pete Austin \orcidID{0000-0003-0238-8662}
    \and
    Sougata Bose   \orcidID{0000-0003-3662-3915}
    \and
    Patrick Totzke \orcidID{0000-0001-5274-8190}  
}
\authorrunning{P.~Austin, S.~Bose, and P.~Totzke}
%
\institute{University of Liverpool, UK}
%
\maketitle              
\begin{abstract}
    Temporal graphs are a popular modelling mechanism for dynamic complex systems that extend ordinary graphs with discrete time. Simply put, time progresses one unit per step and the availability of edges can change with time.

We consider the complexity of solving $\omega$-regular games played on temporal graphs
where the edge availability is ultimately periodic and fixed a priori.

We show that solving parity games on temporal graphs is decidable in \PSPACE, only assuming the edge predicate itself is in \PSPACE.
A matching lower bound already holds for what we call \emph{punctual} reachability games on static graphs, where one player wants to reach the target at a given, binary encoded, point in time.
We further study syntactic restrictions that imply more efficient procedures.
In particular, if the edge predicate is in $\P$ and is monotonically increasing for one player and decreasing for the other, then
the complexity of solving games is only polynomially increased compared to static graphs.

\end{abstract}

\section{Introduction}
\label{sec:intro}
Temporal graphs
are graphs where the edge relation changes over time. They are often presented as a sequence $G_{0},G_{1},\ldots$ of graphs over the same set of vertices.
We find it convenient to define them as pairs $G=(\nodes,\edges)$ 
consisting of a set $\nodes$ of vertices
and associated edge availability predicate 
$\edges:\nodes^{2}\to 2^\N$
that determines at which integral times a directed edge can be traversed.
This model has been used to analyse dynamic networks and distributed systems in dynamic topologies, such as gossiping and information dissemination \cite{ravi1994,KLO2010}.
There is also a large body of work that considers temporal generalisations
of various graph-theoretic notions and properties  \cite{MCS2014,FMS2009,DFS2023}.
Related algorithmic questions include graph colouring \cite{MMZ21},
exploration \cite{EHK21},
travelling salesman \cite{MS14},
maximum matching \cite{MMN23},
and vertex-cover \cite{AMGZ20}.
The edge relation is often
deliberately left unspecified and sometimes only assumed to
satisfy some weak assumptions about connectedness, frequency, or fairness
to study the worst or average cases in uncontrollable environments.
Depending on the application, one distinguishes between ``online'' questions, where the edge availability is revealed stepwise, as opposed to the ``offline'' variant where all is given in advance.
We refer to \cite{HJ19,M2015} for overviews of temporal graph theory and its applications.

Two player zero-sum verification games on directed graphs play a central role in formal verification, specifically the reactive synthesis approach \cite{PR1989}.
Here, a controllable system and an antagonistic environment are modeled as a game in which two opposing players jointly move a token through a graph.
States are either owned by \PONE (the system) or \PTWO (the environment), and the owner of the current state picks a valid successor.
Such a play is won by \PONE if, and only if, the constructed path
satisfies a predetermined \emph{winning condition} 
that models the desired correctness specification.
The winning condition is often given either in a temporal logic such as Linear Temporal Logic (LTL) \cite{P1977},
or directly as $\omega$-automaton whose language is the set of infinite paths considered winning for \PONE.
The core algorithmic problem is solving games: to determine
which player has a strategy to force a win, and if so, how.

Determining the complexity of solving games on static graphs has a long history and continues to be an active area of research.
We refer to \cite{GAMES2002,GAMES2023} for introductions on the topic
and recall here only that
solving reachability games, where \PONE aims to eventually reach a designated target state, is complete for polynomial time.
The precise complexity of solving parity games is a long-standing open question.
It is known to be in $\UP\cap\co\UP$ \cite{J1998}, and so in particular in $\NP$ and $\co\NP$,
and recent advances have led to quasi-polynomial time algorithms
\cite{CJKLS2017,JL2017,LB2020,CF2019,LPSD2022}.

\paragraph{Related Work.}

Periodic temporal graphs were first studied by Floccchini, Mans, and Santoro in \cite{FMS2009}, where they show polynomial bounds on the length of explorations (paths covering all vertices).
Recently, De Carufel, Flocchini, Santoro, and Simard \cite{DFS2023} study Cops \& Robber games on periodic temporal graphs.
They provide an algorithm for solving one-cop games that is only quadratic in the number of vertices and linear in the period. 
%

Games on temporal graphs with maximal age, or period of some absolute value $K$ given in binary are games on exponentially succinctly presented arenas.
Unfolding them up to time $K$ yields an ordinary game on the exponential sized graph which allows to transfer upper bounds, that are not necessarily optimal.
In a similar vein, Avni, Ghorpade, and Guha \cite{AGG2023} have recently introduced types of games on exponentially succinct arenas called pawn games. Similar to our results, their findings provide improved \PSPACE\ upper bounds for reachability games.

Parity games on temporal graphs are closely related to timed-parity games,
which are played on the configuration graphs of timed automata \cite{AD1994}.
However, the time in temporal graphs is discrete as opposed to the continuous time semantics in timed automata.
Solving timed parity games is complete for \EXP \cite{MPS1995,CHP2011}
and the lower bound already holds for reachability games on timed automata with only two clocks \cite{JT2007}.
Unfortunately, a direct translation of (games on) temporal graphs to equivalent timed automata games requires at least two clocks: one to hold the global time used to
check the edge predicate and one to ensure that time progresses one unit per step.

\paragraph{Contributions.}
We study the complexity of solving parity games on temporal graphs.
As a central variant of independent interest are what we call \emph{punctual} reachability games, that are played on a static graph and player wants to reach a target vertex at a given binary encoded time.
We show that solving such games is already hard for \PSPACE, which provides a lower bound for all temporal graph games we consider.

As our second, and main result, we show how to solve parity games on (ultimately) periodic temporal graphs.
The difficulty to overcome here is that the period may be exponential in the number of vertices
and thus a na\"ively solving the game on the unfolding only yields algorithms in exponential space.
Our approach relies on the existence of polynomially sized summaries that can be verified in \PSPACE\ using punctual reachability games.

We then provide a sufficient syntactic restriction that avoids an increased complexity for game solving.
In particular, if the edge predicate is in polynomial time and is monotonically increasing for one player and decreasing for the other, then
the cost of solving reachability or parity games on temporal graphs
increases only polynomially in the number of vertices compared to the cost of solving these games on static graphs.

\smallskip
None of our upper bounds rely on any particular representation of the edge predicate.
Instead, we only require that the representation ensures that checking membership (if an edge is traversable at a given time) has suitably low complexity.
That is, our approach to solve parity games only requires that the edge predicate is in \PSPACE, and polynomial-time verifiable edge predicates suffice
to derive \P-time upper bounds for monotone reachability games.
These conditions are met for example if the edge predicate is defined as semilinear set given as an explicit union of linear sets (\NP\ in general and in \P\ for singleton sets of periods), or by restricted Presburger formulae: the quantifier-free fragment is in \P, the existential fragment is in \NP\ but remains in \P\ if the number of variables is bounded \cite{S1984}. See for instance \cite{H2018} and contained references.

\smallskip
The rest of the paper is structured as follows.
We recall the necessary notations in 
\cref{sec:preliminaries} and then discuss reachability games in \cref{sec:reachability}.
\Cref{sec:parity} presents the main construction for solving parity games and finally,
in \cref{sec:monotone}, we discuss improved upper bounds for monotone temporal graphs.

\section{Preliminaries}
\label{sec:preliminaries}
\newcommand{\timeprop}{\lambda}

\newcommand\prob[3] {
	\noindent
	\medskip
	\smallskip
	\noindent\\{\bfseries Problem:} #1
	\smallskip
	\hrule
	\smallskip
	\noindent{\bfseries Input:} #2\\
	{\bfseries Question:} #3
	\smallskip
	\hrule
}

\begin{definition}[Temporal Graphs]
	A temporal graph $G=(\nodes,\edges)$ is a directed graph where $\nodes$ are vertices and $\edges:\nodes^2\to 2^\N$ is the edge availability relation
        that maps each pair of vertices to the set of times at which the respective directed edge can be traversed.
        If $i\in \edges(s,t)$ we call $t$ an \emph{$i$-successor} of $s$ and write
$s\step{i}t$.

        The \emph{horizon} of a temporal graph
        is $\hor(G)=\sup_{s,t\in\nodes}(\edges(s,t))$, the largest finite time at which any edge is available, or $\infty$ if no such finite time exists.
        A temporal graph is \emph{finite} if $\hor(G)\in\N$
        i.e., every edge eventually disappears forever.
        A temporal graph
        is \emph{periodic} with period $\period\in\N$ if for all nodes $s,t\in\nodes$ it holds that $\edges(s,t) =\edges(s,t) +\period\cdot\N$. We call $G$ \emph{static} if it has period $1$.
\end{definition}

Naturally, one can unfold a temporal graph into its \emph{expansion}
up to some time $T\in\N\cup\{\infty\}$,
which is the graph with nodes $\nodes\x\{0,1,\ldots,T\}$
and directed edges $(s,i)\to (t,i+1)$ iff $i\in\edges(s,t)$.

In order for algorithmic questions to be interesting, we assume that temporal graphs are
given in a format that is more succinct than the expansion
up to their horizon or period.
We only require that the representation ensures that checking if an edge is traversable at a given time can be done reasonably efficiently.

We will henceforth use formulae in the existential fragment of Presburger arithmetic,
the first-order theory over natural numbers with equality and addition.
That is, the $\exists$PA formula $\Phi_{s,t}(x)$ with one free variable $x$
represents the set of times at which an edge from $s$ to $t$ is available
as 
\mbox{$\edges(s,t) = \{n \mid \Phi_{s,t}(n) \equiv \true\}$}.
We use common syntactic sugar including inequality and multiplication with (binary encoded) constants.
For instance, $\Phi_{s,t}(x) \eqdef 5 \le x \land x \le 10$
means the edge is available at times $\{5,6,7,8,9,10\}$;
and $\Phi_{s,t}(x) \eqdef \exists y. (x=y\cdot 7) \land \lnot (x\le 100)$
means multiples of $7$ greater than $100$.

\begin{definition}[Parity Games]
    A \emph{parity game} is a zero-sum game played by two opposing players on a directed graph.
Formally, the game is given by a game graph $G=(\states,\edges)$,
a partitioning $\states=\states_1\uplus\states_2$ of vertices into those owned by \PONE and \PTWO respectively, and a colouring 
$\col:\states\to\colours$
of vertices into a finite set $\colours\subsetneq\N$ of colours.

The game starts with a token on an initial vertex $\init\in\states$
and proceeds in turns where in round $i$, the owner of the vertex occupied by the token moves it to some successor. This way both players jointly agree on an infinite path
$\rho=s_0s_1\ldots$
called a \emph{play}.
A play is winning for \PONE if 
$\max\{c\mid \forall i \exists j. \col(s_j)=c\}$, the maximum colour seen infinitely often,
is even.

A \emph{strategy} for Player~$i$ is a recipe for how to move.
Formally, it is a function $\sigma_i:\states^*\states_i\to \states$
from finite paths ending in a vertex $s$ in $V_i$ 
to some successor.
%
We call $\sigma$ \emph{positional} 
if 
$\sigma(\pi s) = \sigma(\pi' s)$
for any two prefixes $\pi,\pi'\in\states^*$.
A strategy is \emph{winning from vertex $s$} if Player~$i$ wins every play
that starts in vertex $s$ and during which all decisions are made according to $\sigma$. 
\end{definition}

We call a vertex $s$ winning for Player~$i$ if there exists a winning strategy from $s$, and call the subset of all such vertices the \emph{winning region} for that player.
Parity games enjoy the following property (See \cite[Theorem~15]{GAMES2023} for details).
\begin{proposition}
\label{thm:uprop}
    Parity games are uniformly positionally determined:
    For every game $(\statespart,\edges,\col)$
    there is a pair $\sigma_1,\sigma_2$ of positional strategies
    so that $\sigma_i$ is winning for Player~$i$
    from every vertex in the winning region of Player~$i$.
%
\end{proposition}


\medskip
A \emph{temporal parity game} is a parity game played on the infinite expansion of a temporal graph $G=(\states,\edges)$,
where the ownership and colouring of vertices are given
with respect to the underlying directed graph
$\statespart$ 
and $\col:\states\to\colours$.
The ownership and colouring are lifted to the expansion so that
vertices in $\states_i\x\N$ are owned by Player~$i$
and vertex $(s,n)$ has colour $\col(s)$.

\begin{example}\label{ex:prg}
Consider the temporal parity game shown in \cref{fig:tpg}.
We will draw \PONE states as diamond and those controlled by \PTWO as squares
and sometimes write modulo expressions to define the edge availability.
For example, 
the constraint on the edge from $u$ to $v$ can be written as the $\exists$PA-formula as $\exists y. (x= 3y) \lor (x=3y+1)$
and so this edge is available at times $0,1,3,4,6,\dots$.
The temporal graph underlying this game has period $15$.

\PONE has a winning strategy starting from $(s,i)$ in 
the expansion by staying in state $s$ until time  $i'\ge i$ with $i'\equiv 0 \mod 5$  and then following the edge to $(t,i'+1)$.
If \PTWO ever chooses to move to $r$, he is trapped in an even-coloured cycle;
if he stays in $t$ forever, then the resulting game sees only colour $2$ and is losing for him.
Otherwise, if the game continues at $(s,i'+2)$, \PONE repeats as above (and wins plays that see both states $s$ and $t$.
The example shows that \PONE s strategies depend on the time and are not positional in the vertices alone, even if the winning set has period $1$.
Indeed, the only possible vertex-positional strategy (cycle in $s$) is losing.

The vertices $\{s,t\}$ shaded in blue represent the vertex from which \PONE can win starting at any time, following the strategy described above.
From the vertices shaded in red, \PTWO can win starting at certain times. For example,  \PTWO has a winning strategy from $(u,i)$ if, and only if, $i\equiv 0 \mod 3$ or $i\equiv 1 \mod 3$ by moving to $(v,i+1)$. Notice that this edge is not available, and thus \PTWO is forced to move to $t$ at times $x\equiv2 \mod 3$. In particular therefore, \PONE wins from $(v,0)$.
The winning region for \PONE is $\{(s,k), (t,k), (r, k), (u,3k+2), (v,3k), (w,3k+1)~\mid ~ k\in \N\}$.
\end{example}

\begin{figure}[t]
	\begin{center}
                    \begin{tikzpicture}[
RED/.style = {color=red, fill=red!20},
BLUE/.style = {color=blue, fill=blue!20}
    ]
        \node[estate, BLUE, label={1}] (s) at (0, 0){$s$};
        \node[estate,RED, label={3}] (v) at (8, 0){$v$};
        \node[ostate,RED, label={2} ] (u) at (4, 0){$u$};
        \node[ostate,RED, label={below:2}](w) at (8, -3){$w$};
        \node[ostate,BLUE, label={below:2}] (t) at (0, -3){$t$};
        \node[ostate,BLUE, label={below:4}] (r) at (4, -3){$r$};

        \draw  (s) edge [loop left] node [above]{} (s);
        \draw  (t) edge [loop left] node [above]{} (t);
        \draw  (s) edge (u);
        \draw  (s) edge [bend right] node [sloped,auto]{$x\equiv0 \mod 5$} (t);
        \draw  (t) edge [bend right] (s);
        \draw  (u) edge node [above, sloped]{$\lnot (x\equiv 0\mod3)$} (t);
        \draw  (u) edge node [above]{$(x\equiv0 \mod 3)\lor$}
                        node [below]{$(x\equiv 1 \mod 3)$} (v);
        \draw  (v) edge (w);
        \draw  (w) edge (u);
        \draw (t) edge (r);
        \draw (r) edge[loop right] (r);
\end{tikzpicture}
	\end{center}
	\caption{
An example of a temporal parity game. \PONE controls the diamond vertices $V_1= \{s,v\}$ and \PTWO controls square vertices $V_2=\{r,t,u,w\}$.
Edge labels are Presburger formulae constraints denoting when an edge is available; edges without constraints are always available.
The grey label next to each node denotes its colour.
E.g.,
$\col(s)=1 \in \colours=\{1,2,3,4\}$.
}
	\label{fig:tpg}
\end{figure}
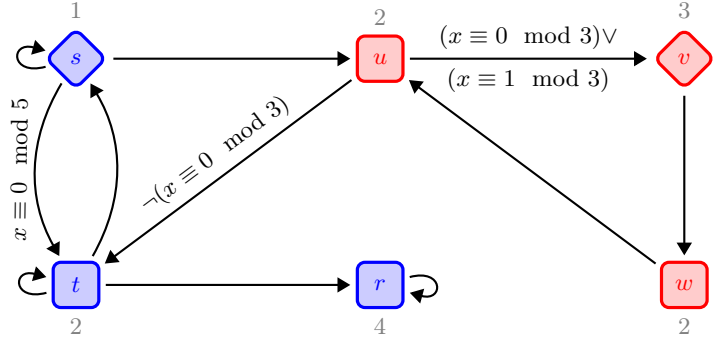

The algorithmic question we consider is determining the set of vertices from which \PONE wins starting at time $0$.

\section{Reachability Games}
\label{sec:reachability}
We discuss a variant of temporal games that turns out to be central both for upper and lower bounds for solving games on temporal graphs.

We call these \emph{punctual reachability games},
which are played on a static graph and \PONE aims to reach the target precisely at a target time.

\begin{definition}
    A \emph{punctual} reachability game 
    $G=(\states,\edges,\init,F)$
    is a game played on a static graph with vertices $\states=\states_1\uplus\states_2$, edges $\edges\subseteq\states^{2}$,
    an initial state $s_0$ and set of target vertices $F\subseteq \states$.
    An additional parameter is a target time $T\in\N$ given in binary.
    \PONE wins a play if and only if a vertex in $F$ is reached at time $T$.
\end{definition}

Punctual reachability games are really just a reformulation
of
the membership problem for alternating finite automata (AFA) 
\cite{CKS1981}
over a unary input alphabet.
\PONE wins the punctual reachability game with target $T$
if, and only if, the word $a^T$ is accepted by the AFA described by the game graph.
Checking if a given unary word $a^T$ is accepted by an AFA is complete for polynomial time if $T$ is given in unary \cite{JR1991}.
We first observe that it is \PSPACE-hard if $T$ is given in binary.
We write in the terminology of punctual reachability games
but the main argument is by reduction from the emptiness problem for unary AFA, which is \PSPACE-compete \cite{H1995,JS2007}.
We rely on the fact that the shortest word accepted by an AFA is at most exponential in the number of states.

\begin{lemma}
    \label{lem:PRG-SWP}
    Let $G=(\states,\edges,\init,F)$
    be a reachability game on a static graph.
    If there exist $T\in\N$ so that \PONE wins the punctual reachability game at target time $T$,
    then there exists some such $T\le 2^{\card{\states}}$.
\end{lemma}
\begin{proof}
    Assume towards contradiction that $T\ge 2^{\card{\states}}$ is the smallest number such that \PONE wins the punctual reachability game and consider some winning strategy $\sigma$.
    For any time $k\le T$ we can consider the set $S_k\subseteq \states$ of vertices occupied on any branch of length $k$ on $\sigma$.
    By the pigeonhole principle, we observe $k<k'\le T$ with $S_k=S_{k'}$,
    which allows to create a strategy $\sigma'$ that follows $\sigma$ until time $k$,
    then continues (and wins) according to $\sigma$ as if it had just seen
    a length $k'$ history leading to the same vertex. This shows that there exists a winning strategy for target time $T-(k-k')$, which contradicts the assumption.  
    \qed
\end{proof}

A lower bound for solving punctual reachability games is now immediate.

\begin{lemma}
    \label{lem:reachability-LB}
    Solving punctual reachability games with target time $T$ encoded in binary
    is \PSPACE-hard.
\end{lemma}
\begin{proof}
    We reduce the non-emptiness problem
    of AFA over unary alphabets. In our terminology this is the decision problem
    if, 
    for a given a reachability game
    $G=(\states,\edges,\init,F)$
    there exists some $T\in\N$ so that \PONE wins the punctual reachability game at target time $T$.
    This problem is \PSPACE-complete \cite{H1995}.

    By \cref{lem:PRG-SWP}, positive instances can be witnessed by a small target $T\le 2^{\card{\states}}$ and so we know that it is \PSPACE-hard to determine the existence of such a small target time that allows \PONE to win.
    
    Consider now the punctual 
    reachability game $G'$ that extends $G$ by a new
    initial vertex $\init'$ that is owned by \PONE and which has 
    a self-loop as well as an edge to the original initial vertex $\init$
    with target time $T'\eqdef 2^{\card{\states}}$.
    In $G'$, \PONE selects some number $T\le T'$ by waiting in the initial vertex for $T'-T$ steps and then starts the game $G$ with the target time $T$.
    Therefore, \PONE wins in $G'$ for target $T'$ if, and only if, she wins for some $T\le2^{\card{\states}}$ in $G$.
    \qed
\end{proof}
\begin{corollary}
    \label{lem:TRG-PSPACE-hard}
    Solving reachability games on finite temporal graphs is
    \PSPACE-hard.
\end{corollary}
\begin{proof}
    We reduce the punctual reachability game with target $T$ to an ordinary reachability game on a finite temporal graph.
    This can be done by introducing a new vertex $u$ as the only target vertex,
    so that it is only reachable via edges from vertices in $F$ at time exactly $T$.
    That is $\edges(s,u) \eqdef \{T\}$
    and $\edges(s,t)=[0,T]$ for all $s,t\in\states\setminus\{u\}$.
    Now \PONE wins the reachability game for target $u$ if, and only if, she wins the punctual reachability game with target $F$ at time $T$.   
    \qed
\end{proof}


A matching \PSPACE\ upper bound for solving punctual reachability games, as well as reachability games on finite temporal graphs can be achieved by computing the winning region backwards as follows.\footnote{For readers familiar with reachability games, the notion $\Pre[1](S)$ above is very similar to, but not the same as 
the $k$-step attractor of $S$: The former contains states from which \PONE can force to see the target in \emph{exactly} $k$ steps, whereas the latter
contains those where the target is reachable in $k$ \emph{or fewer} steps.}
For any game graph with vertices $\statespart$, set $S\subseteq\states$ and $i\in\{1,2\}$, let
$\Pre[i](S)\subseteq V$ denote the set of vertices from which Player~$i$ can force to reach $S$ in one step.
\[
    \Pre[i](S)\eqdef 
\{v\in\states_i\mid \exists (v,v')\in \edges. v'\in S\}
\cup \{v\in\states_{1-i}\mid \forall (v,v')\in \edges. v'\in S\}
\]
A straightforward induction on the duration $T$ shows that
Player~$i$ wins the punctual reachability game with target time $T$ from vertex $s$
if, and only if $s\in\Pre[i]^T(F)$, the $T$-fold iteration of $\Pre[i]$ applied to the target set $F$.

Notice that knowledge of $\Pre[i](S)$ is sufficient to compute 
$\Pre[i]^{k+1}(S)$. 
By iteratively unfolding the definition of $\Pre[i]^{k}$, we can
compute $\Pre[1]^T(F)$ from $\Pre[1]^0(F)=F$
in polynomial space\footnote{To be precise, na\"ively unfolding the definition requires $\?O(T + \card{\states}^2)$ time, exponential in (the binary encoded input) $T$, and $\?O(\card{\states}+\log(T))$ space to memorise the current set $\Pre[k]\subseteq \states$ as well as the time $k\le T$ in binary.}.
Together with \cref{lem:reachability-LB} we conclude the following.

\begin{theorem}
    \label{thm:reachability}
    Solving punctual reachability games with target time $T$ encoded in binary
    is \PSPACE-complete.
\end{theorem}
The same approach works for reachability games on finite \emph{temporal} graphs
if applied to the expansion up to horizon $\hor(G)$,
leading to the same time and space complexity upper bounds.
The only difference is that computing $\Pre^k(F\x\{T\})$ requires
to check edge availability at time $T-k$.

\begin{theorem}
    \label{lem:TRG-PSPACE}
    Solving reachability games on finite temporal graphs is \PSPACE-complete.
\end{theorem}
\begin{proof}
    Consider a temporal game with vertices $\statespart$, edges $\edges:\states^{2}\to2^{\N}$
    target vertices $F\subseteq\states$ and where $T=\hor(G)$ is the latest time an edge is available.
    We want to
    check if starting in an initial state $\init$ at time $0$,
    \PONE can force to reach $F$ at time $T$.
    In other words, for the game played on the expansion up to time $T$
    we want to decide if $(\init,0)$ is contained in $\Pre[1]^T(F\x\{T\})$.

    

    By definition of the expansion,
    we have
    $\Pre(S\x\{n\}) \subseteq \states\x\{n-1\}$ 
    for all $S\subseteq\states$ and $n\le T$.
    Since we can check the availability of an edge at time $n$ in polynomial space,
    we can iteratively compute 
    $\Pre^n(F\x\{T\})$ backwards, starting with $\Pre[1]^0(F\x\{T\}) = F\x\{T\}$,
    and
    only memorising the current iteration $n\le T$ and a set $W_n\subseteq \states$
    representing $\Pre[1]^n(F\x\{T\})=W_n\x\{T-n\}$.
    \qed
\end{proof}



\section{Parity Games}
\label{sec:parity}
\newcommand{\bigstep}{B}
We consider Parity games played on periodic temporal graphs.
As input we take a temporal graph $G=(\states,\edges)$ with period $K$,
a partitioning $\statespart$ of the vertices,
as well as a colouring $\col:\nodes\to\colours$ that associates a colour out of a finite set $\colours\subset\mathbb{N}$ of colours to every state.

It will be convenient to write
$\col(\pi) \eqdef \max\{\col(s_i)\mid 0\le i\le k\}$ for the maximal colour of any vertex visited
along a finite path
$\pi=(s_0,0)(s_1,1)\ldots(s_k,k)$.
%
The following
relations $\ReachCols{\strategy}{s}$ capture the guarantees provided by a strategy $\sigma$ if followed for
one full period from vertex $s$.

\begin{definition}
    For a strategy
    $\strategy$ and 
    vertex $s\in\states$
    define
    $\ReachCols{\strategy}{s}\subseteq \states\x\colours$ 
%
    be the relation containing 
    $(t,c)\in\ReachCols{\strategy}{s}$ if, and only if,
    there exists a finite play $\pi=(s,0)\ldots (t,\period)$ consistent with $\sigma$,
    that starts in $s$ at time $0$, ends in $t$ at time $K$,
    and the maximum colour seen on the way is $\col(\pi)=c$.
    We call $\ReachCols{\strategy}{s}$ 
    the \emph{summary} of $s$ with respect to strategy $\strategy$.

    A relation $B\subseteq \states\x\colours$
    is $s$-\emph{realisable} if
    there is a strategy $\strategy$
    with
    $B=\ReachCols{\strategy}{s}$.
\end{definition}

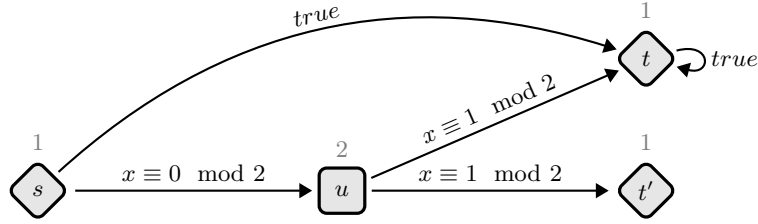
\begin{figure}[t]
    \centering
    \begin{tikzpicture}
    \node[estate, label={1}](source) at (0,0){$s$};
    \node[ostate, label={2}](safety) at (4,0){$u$};
    \node[estate, label={1}](sink) at (8,1.75){$t$};
    \node[estate, label={1}](sinkprime) at (8,0){$t'$};
    
    \path[every loop/.append style=-{>}]
    (source) edge node [above,sloped]{$x\equiv0\mod2$} (safety)
    (source) edge [bend left] node [above,sloped]{$\true$} (sink)
    (safety) edge node [above,sloped]{$x\equiv1\mod2$} (sink)
    (safety) edge node [above,sloped]{$x\equiv1\mod2$} (sinkprime)
    (sink) edge [loop right] node [auto]{$\true$} (sink)
    ;
\end{tikzpicture}
    \caption{The game from \cref{ex:realisable}. Labels on vertices and edges denote colours and available times, respectively.
    The graph has period $2$. In two rounds, \PONE can force to end in $t$ having seen colour $1$, or in either $t$ or $t'$ but having seen a better colour $2$.}
    \label{fig:example1}
\end{figure}
\begin{example}
    \label{ex:realisable}
    Consider the game in \cref{fig:example1}
    where vertex $u\in\states_2$ has colour $2$ and all other vertices have colour $1$.
    The graph has period $K=2$.
    The relations
    $\{(t,1)\}$ and 
    $\{(t,2),(t',2)\}$ are $s$-realisable,
    as witnessed by the strategies $\sigma(s)=t$
    and $\sigma(s)=u)$, respectively.
    However, $\{(t,2)\}$ is not
    $s$-realisable as no \PONE strategy guarantees to visit $s$ then $u$ then $t$.
    %
\end{example}

\begin{lemma}\label{lem:realisability}
    Checking $s$-realisability is in \PSPACE.
    That is, one can verify in polynomial space for a given
    temporal Parity game, state $s\in\states$ and relation
    $B\subseteq \states\x\colours$
    whether $B$ is $s$-realisable.
\end{lemma}
\begin{proof}
    We reduce checking realisability to solving a reachability game on a temporal graph that is only polynomially larger.
    More precisely, given a game $G=(\states,\edges,\col)$
    consider the game $G'=(\states',\edges',\col')$ over vertices $\states'\eqdef\states\x\colours$
    that keep track of the maximum colour seen so far.
    That is, the ownership of vertices and colours are lifted directly as
    $(s,c)\in\states'_1\iff s\in\states_1$ and 
    $\col'(s,c)\eqdef\col(s)$,
    and for any $i\in\N$, $s,t,\init\in\states$, $c,d\in C$,
    we let $(t,d)$ be an $i$-successor of $(s,c)$ if, and only if, both
    $t$ is an $i$-successor of $s$
    and $d=\max\{c,\col(t)\}$.

    Consider some relation $B\subseteq\states\x\colours$.
    We have that $B$ is $s$-realisable
    if, and only if,
    \PONE wins the punctual reachability game on $G'$
    from vertex $(s,\col(s))$ at time $0$, towards target vertices $B\subseteq\states'$
    at target time $\period$.
    Indeed, any winning \PONE strategy in this game witnesses that $B$ is  $s$-realisable and vice versa.
    By \cref{lem:TRG-PSPACE}, the existence of such a winning strategy can be verified in polynomial space
    by backwards-computing the winning region.
    \qed
\end{proof}

The following defines a small, and \PSPACE-verifiable certificate for \PONE to win the parity game on a periodic temporal graph.
\begin{definition}[Certificates]
    Given temporal parity game $(\states,\edges,\col)$ with period $\period$,
    a \emph{certificate} for \PONE winning the game from initial vertex $\init\in\states$
    is a multigraph
    where the vertex set $\states'\subseteq\states$ contains $\init$,
    and edges $\edges'\subseteq\states'\x\colours\x\states'$ are labelled by colours, such that
    \begin{enumerate}
        \item For every $s\in\states'$, the set
            $\Post{s}\eqdef\{(t,c) \mid (s,c,t) \in \edges'\}$
            is $s$-realisable.
        \item The maximal colour on every cycle reachable from $\init$ is even.
    \end{enumerate}
\end{definition}

Notice that condition 1 implies that no vertex in a certificate is a deadlock.
A certificate intuitively allows to derive \PONE strategies based on those witnessing the realisability condition.

\begin{figure}[t]
	\begin{center}
                    \begin{tikzpicture}
        \node[state] (s) at (0, 0){$s$};
        \node[state] (v) at (8, 0){$v$};
        \node[state] (t) at (0, -3){$t$};
        \node[state](r) at (4, -3){$r$};

        \draw  (s) edge [loop left] node {$2$} (s);
        \draw  (t) edge [loop left] node {$2$} (t);
        \draw  (s) edge [bend right] node[left]{$2$} (t);
        \draw  (s) edge  node[auto] {$4$} (r);
        \draw  (t) edge [bend right] node[right] {$2$} (s);
        \draw  (v) edge node[above] {$3$} (s);
        \draw  (v) edge node[above] {$3$} (t);
        \draw  (v) edge node[above] {$4$} (r);
        \draw (t) edge  node[above] {$4$}(r);
        \draw  (r) edge [loop right] node[right] {$4$} (r);
\end{tikzpicture}
	\end{center}
	\caption{A certificate that \PONE wins the game in \cref{ex:prg}
        from state $v$ at time $0$.
        }
	\label{fig:cert}
\end{figure}
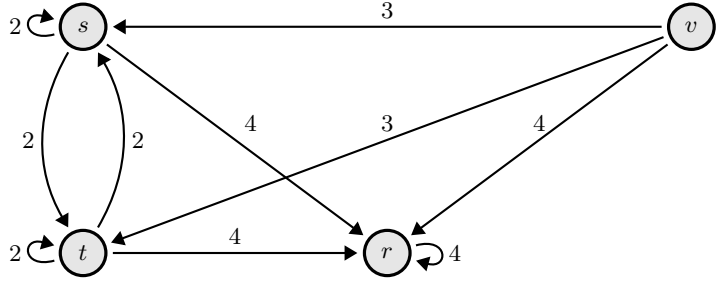

\begin{example}
    Consider the game from \cref{ex:prg} played on the temporal graph with period 15. A certificate for \PONE winning from state $v$ at time $0$
    is depicted in \cref{fig:cert}.
    Indeed, the \PONE strategy mentioned in \cref{ex:prg} (aim to alternate between $s$ and $t$) witnesses that 
    $\Post{v} = \{(s,3), (t,3), (r,4)\}$ is $v$-realisable
    because it allows \PONE to enforce that after $K=15$ steps from $v$,
    the game ends up in one of those states via paths whose colour is dominated by $\col(v)=3$ or $\col(r)=4$.
\end{example}

\begin{lemma}\label{lem:par-correctness}
   \PONE wins the parity game on $G$ from vertex $\init$ 
   if, and only if, there exists a certificate.
\end{lemma}
\begin{proof}
    For the backward implication we argue that a certificate $C$
    allows to derive a winning strategy for \PONE in the parity game $G$.
    By the realisability assumption (1),
    for each vertex $s\in\states$ there must exist a \PONE strategy $\sigma_s$
    with $\ReachCols{\sigma_s}{s}=\Post{s}$ that tells her how to play in $G$
    for $\period$ rounds if the starting time is a multiple of $\period$.
    Moreover, suppose she plays according to $\sigma_s$ for $\period$ rounds
    and let $t$ and $c$ be the vertex reached and maximal colour seen on the way.
Then by definition of the summaries, $(t,c)\in\ReachCols{{\sigma_s}}{s}=\Post{s}$
    and so in the certificate $C$ there must be some edge $s\step{c}t$.
    
    Suppose \PONE continues to play in $G$ like this forever:
    From time $i\cdot\period$ to $(i+1)\cdot\period$ she plays
    according to some strategy $\strategy_{s_{i}}$
    determined by the vertex $s_i$ reached at time $i\cdot \period$.
    Any consistent infinite play $\rho$ in $G$, chosen by her opponent,
    describes an infinite walk $\rho'$ in $C$ such that the colour seen in 
    any step $i\in\N$ of $\rho'$ is precisely the dominant colour 
    on $\rho$ between rounds $i\period$ and $(i+1)\period$.
    Therefore the dominant colours seen infinitely often on $\rho$ and $\rho'$ are the same and, by certificate condition (2) on the colouring of cycles, even.
    We conclude that the constructed strategy for \PONE is winning.

    For the forward implication, assume that \PONE wins the game on $G$ from vertex $s$ at time $0$. 
    Since the game $G$ is played on a temporal graph with period $\period$,
    its expansion up to time $\period-1$
    is an ordinary parity game
    on a static graph with vertices $\nodes\x\{0,1,\ldots,\period-1\}$ where the second component indicates the time
    modulo $\period$.
    Therefore, by positional determinacy of parity games (\cref{thm:uprop}), we can assume that
    \PONE wins in $G$ using a strategy $\strategy$ that is itself periodic. That is,
    $\strategy(hv) = \strategy(h'v)$
    for any two histories $h,h'$ of lengths $\length{h}\equiv \length{h'} \mod \maxtime$.
    Moreover, we can safely assume that $\sigma$ is uniform, meaning that it is winning from any vertex $(s,0)$ for which a winning strategy exists.
    Such a strategy induces a multigraph $C=(\states,\edges')$
    where the edge relation is defined by 
    $(s,c,t)\in \edges' \iff (t,c)\in\ReachCols{\strategy}{s}$.
    It remains to show the second condition for $C$ to be a certificate,
    namely that any cycle in $C$, reachable from the initial vertex $\init$, has an even maximal colour.
    Suppose otherwise, that $C$ contains a reachable cycle whose maximal colour is odd.
    Then there must be play in $G$ that is consistent with $\strategy$
    and which sees the same (odd) colour infinitely often. 
    But this contradicts the assumption that $\strategy$ was winning in $G$ in the first place.
    \qed
\end{proof}

Our main theorem is now an easy consequence of the existence of small certificates.

\begin{theorem}
    \label{thm:parity}
    Solving parity games on periodic temporal graphs is 
    \PSPACE-complete.
\end{theorem}
\begin{proof}
    Hardness already holds for reachability games \cref{lem:reachability-LB}.
    For the upper bound we show membership in \NPSPACE\ and use Savitch's theorem.
    By \cref{lem:par-correctness} it suffices to guess and verify a candidate certificate $C$.
    These are by definition polynomial in the number of vertices and colours in the given
    temporal parity game.
    Verifying the cycle condition (2) is trivial in polynomial time
    and 
    verifying the realisability condition (1) is in \PSPACE\ by \cref{lem:realisability}.
    \qed
\end{proof}

\begin{remark}
    \label{rem:parity-ultimately-periodic}
    The \PSPACE\ upper bound in \cref{thm:parity} can easily be extended to games on temporal graphs that are \emph{ultimately} periodic, meaning that there exist $T,\period\in\N$ so that for all $n\ge T$,
    $s\step{n}t$ implies $s\step{n+K}t$.
    Such games can be solved by first considering the periodic suffix according to \cref{thm:parity}
    thereby computing the winning region for \PONE at time exactly $T$,
    and then solving the temporal reachability game with horizon $T$.
\end{remark}

\section{Monotonicity}
\label{sec:monotone}
In this section, we consider the effects of monotonicity assumptions on the edge relation with respect to time on the complexity of solving reachability games. 
We first show that reachability games remain \PSPACE-hard even if the edge relation is decreasing (or increasing) with time. We then give a fragment for which the problem becomes solvable in polynomial time.

\paragraph*{Increasing and Decreasing temporal graphs: }
Let the edge between vertices \\$u,v\in\states$ of a temporal graph be referred to as
\emph{decreasing} 
if 
$u\step{i+1}v$ implies $u\step{i}v$ for all $i\in \N$, i.e. edges can only disappear over time.
Similarly, call the edge \emph{increasing}
if for all $i\in\N$ we have
that $u\step{i}v$ implies $u\step{i+1}v$; i.e. an edge available at current time continues to be available in the future.
A temporal graph is decreasing (increasing) if all its edges are.
We assume that the times at which edge availability changes are given in binary. More specifically, every edge is given as inequality constraint
of the form $\Phi_{u,v}(x)\eqdef x\leq n$ (respectively $x\geq n$) for some $n\in \N$. 

Although both restrictions imply that the graph is ultimately static,
we observe that solving reachability games on 
such monotonically increasing or decreasing temporal graphs
remains \PSPACE-complete.

\begin{theorem}
    \label{thm:monotone-decreasing}
	Solving reachability and Parity games on 
	decreasing (respectively increasing)
	temporal graphs is \PSPACE-complete.
\end{theorem}

\begin{figure}[t]
	\centering
	\resizebox{0.8\textwidth}{!}{
					\begin{tikzpicture}
				
				\node[ostate] (v) at (10, 0){$v$};
				\node[estate, color=red, fill=red!20] (w) at (12.75, 0){$w$};
				\node[estate, color=red, fill=red!20] (top) at (15, -1){$\top$};
				\node[estate, color=red, fill=red!20] (bot) at (15, 1){$\bot$};
				
				\path[every loop/.append style=-{>}]
				(7, 1.5) edge [-] node [above]{} (10.08, 1.5)
				(7, -1.5) edge [-] node [above]{} (10.08, -1.5)
				(10, 1.56) edge [-] node [above]{} (10,0.25)
				(10, -1.56) edge [-] node [above]{} (10,-0.25)
				(v) edge [color=red] node [above]{} (w)
				(v) edge [color=red, bend left] node [above, sloped]{$x\leq T-1$} (bot)
				(w) edge [color=red] node [above, color=red, sloped]{$x\leq T+1$} (top)
				(w) edge [color=red] node [above]{} (bot)
				(bot) edge [loop right, color=red] node [above]{} (bot)
				(top) edge [loop right, color=red] node [above]{} (top)
				;	
			\end{tikzpicture}
	}
	\caption{Reduction from a punctual reachability game to a
		reachability game on a temporal graph that is finite and decreasing, see \cref{thm:monotone-decreasing}.
		Components added are shown in red.
	}
	\label{FigureTwo}
\end{figure}
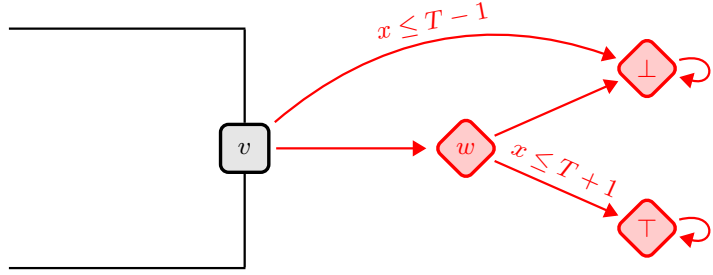

\begin{proof}
    The upper bound holds for parity games as the description of the temporal graph explicitly includes a maximal time $T$ from which the graph becomes static.
    One can therefore solve the Parity game for the static suffix graph (in \NP) and then apply the
    \PSPACE\ procedure (\cref{lem:TRG-PSPACE})
    to solve for temporal reachability towards the winning region at time $T$.
    Alternatively, the same upper bound also follows from 
    \cref{thm:parity,rem:parity-ultimately-periodic}.
    
    For the lower bound we reduce from punctual reachability games which are \PSPACE-hard by 
    \cref{lem:reachability-LB}.
    Consider a (static) graph $G$ and a target time $T\in \N$ given in binary. 
    Without loss of generality, 
    assume that the target vertex $v$ has no outgoing edges.
	We convert $G$ into a temporal graph $G'$ with $V'= V\cup \{w,\top,\bot\}$,
	$V'_1= (V_1\setminus\{v\} )\cup \{w\} $, $V'_2= V'\setminus V'_1$ and new target $\top$.
	The vertex $\bot$ is a sink state
        and the original target vertex $v$ is now controlled by \PTWO.
        Edge availabilities are 
        $v\step{x}\bot $ if $x\leq T-1$,
        $v\step{x}w$ if $x\leq T+1$,
        $w\step{x}\top $ if $x\leq T+1$, 
        and all other edges disappear after time $T+1$.
        The constructed temporal graph is finite and decreasing.
        See \cref{FigureTwo}.
	The construction ensures that the only way to reach $\top$ is to reach $v$ at time $T$, $w$ at time $T+1$ and take the edge from $w$ to $\top$ at time $T+1$. 
	\PONE wins in $G'$ if and only if she wins the punctual reachability game on $G$. 
	
	A similar reduction works in the case of increasing temporal graphs 
	by switching the ownership of vertices $v$ and $w$. 
	The vertex $v$, now controlled by \PONE has the edge $v\step{x}w$ at times 
	$x\ge T$ and the edge $v\step{}\bot$ at all times. 
	The vertex $w$ now controlled by \PTWO has the edge $w\step{}\top$ available at all times but the edge $w\step{x}\bot$ becomes available at time $x\ge T+2$. 
	\qed
\end{proof}

\paragraph*{Declining and improving temporal games: }
We now consider the restriction where all edges controlled by one player are increasing and those of the over player are decreasing.
Taking the perspective of the system \PONE,
we call a game on a temporal graph \emph{declining}
if all edges $u\step{}v$ with $u\in\states_1$ are decreasing
and 
all edges $u\step{}v$ with $u\in\states_2$ are increasing.
Note that \emph{declining} is a property of the game and not the graph as the definition requires a distinction based on ownership of vertices, which
is specified by the game and not the underlying graph.
From now on, we refer to such games as declining temporal reachability (or parity) games.
Notice
that \PONE has fewer, and \PTWO has more choices to move at later times.
Analogously, call the game \emph{improving}
if
the conditions are reversed, i.e.,
all edges $u\step{}v$ with $u\in\states_1$ are increasing
and 
all edges $u\step{}v$ with $u\in\states_2$ are decreasing.

We show that declining (and improving) temporal reachability games can be solved in polynomial time.

\begin{theorem}
	Solving declining (respectively improving) temporal reachability games  is in \P.
\end{theorem}

\begin{proof}
	We first give the proof for declining games. 
	Consider the reachability game on the expansion with vertices $V\times \N$ such that the target set is $F\times \N$. 
	For $k\in\N$ let $W_k\subseteq V$ be the set of those vertices $u$
        such that \PONE has a winning strategy from $(u,k)$.
	We first show that

        \begin{equation}
            \label{eq:mono-sets}
        W_{i+1}\subseteq W_i
        \end{equation}

	For sake of contradiction, suppose there exists 
	$u\in W_{i+1}\setminus W_i$.
	Let $\sigma_{i+1}^1$ be a (positional) winning strategy from $(u,i+1)$ for \PONE in the expansion.
	Since $u \not \in W_i$, by positional determinacy of reachability games
        (\cref{thm:uprop}), \PTWO has a winning strategy $\sigma_{i}^2$ from $(u,i)$.
	Consider a strategy $\sigma_{i}^1$ for \PONE, such that for all $v\in V_1$, $\sigma_{i}^1(v,k)\eqdef \sigma_{i+1}^1(v,k+1)$, for all $k\ge i$. 
	Similarly, let $\sigma_{i+1}^2$ be the strategy for \PTWO, 
	such that for all $v\in V_2$,  
	$\sigma_{i+1}^2(v,k+1)=\sigma_{i}^2(v,k)$, for all $k\ge i$, 
	Note that this is well defined because by definition of declining games, i.e, $v\step{k+1}u$ implies $v\step{k}u$ for all $v\in V_1$, and $v\step{k}u$ implies $v\step{k+1}u$, for all $v\in V_2$.
	Starting from the vertex $(u,i+1)$, the pair of strategies $(\sigma_{i+1}^1,\sigma_{i+1}^2)$ defines a unique play $\pi_{i+1}$, which is winning for \PONE. Similarly, the pair of strategies 
	$(\sigma_{i}^1,\sigma_{i}^2)$ define a play $\pi_{i}$ which is winning for \PTWO starting from $(u,i)$. However, the two plays visit the same set of states, particularly, $(v,k)$ is visited in $\pi_i$ if and only if $(v,k+1)$ is visited in $\pi_{i+1}$. Therefore, either both are winning for \PONE or both are losing for \PTWO, which is a contradiction.
	\smallskip
	Let $N\subseteq \N$ be the set of times at which the graph changes, i.e.
        \[
        N=\{c~\mid~ \exists \Phi_{u,v}(x)=x \triangleleft c\text{, where } \triangleleft\in \{\le, \ge\}\}\}
        \]
    Let $m\eqdef\max(N)$ be the latest time any edge availability changes.
	We show that $W_m= W_k$ for all $k\ge m$. To see this, note that $W_m$ is equal to the winning region for \PONE in the (static) reachability game played on $G_m=(V, E_m)$, where $E_m=\{(u,v)~\mid~ u\step{m}v\}$. Consider a (positional) winning strategy $\sigma_m$ for \PONE in $G_m$ and define a positional strategy $\sigma(v,k)=\sigma_m(v)$, for $k\ge m$. Since the graph is static after time $m$, this is well defined.	
	Starting from a vertex $(u,k)$, a vertex $(v,k+k')$ is visited on a $\sigma$-consistent path if and only if there is a $\sigma_m$-consistent path $u\step[k']{}v$. Therefore, $\sigma$ is a winning strategy 
	from any vertex $(v,k)$ such that $k\ge m$ and $v\in W_m$.
	Moreover, the set $W_m$ can be computed in time $\mathcal{O}(|V|^2)$ by solving the reachability game on $G_m$ \cite[Theorem~12]{GAMES2023}.

	\begin{algorithm}[t]
		\caption{Algorithm for declining games with set of change times $N$ and $m=\max(N)$}\label{alg:ptime}
		
		\begin{algorithmic}
			\State $W\gets\text{Solve}(G_m)$   \Comment{Computes \PONE winning region  in $G_m$}
			\While{$N\neq \emptyset$}
			\State $n\gets max(N)$
			\If{ $(\Pre(W\times \{n\}) = W$}	
			\State $N \gets N\setminus n$ \Comment{Accelerate to next change time}
			\Else	
			\State $W\gets \Pre(W)$ 
			\State $N\gets N\cup \{n-1\}\setminus \{n\}$ 
			\EndIf
			\EndWhile
			
		\end{algorithmic}
		
	\end{algorithm}

	To solve reachability on declining temporal games,
	we can first compute the winning region $W_m$ in the stabilised game $G_m$. This means $W_m\times [m,\infty)$ is winning for \PONE.
	To win the declining temporal reachability game, \PONE can play the punctual reachability game with target set $W_m$ at target time $m$.
	The winning region for \PONE at time $0$ can therefore be computed as $\Pre^m(W_m\times \{m\})$ as outlined in the proof of \cref{lem:TRG-PSPACE}.
	Note that na\"ively this only gives a \PSPACE\ upper bound
        as in the worst case, we would compute $\Pre$ an exponential ($m$) times.
	
		To overcome this, note that in the expansion graph
        $\Pre^i(W_m \times \{m\})= W_{m-i} \times \{m-i\}$.
        According to
        \cref{eq:mono-sets},
        $W_{m-i}\subseteq W_{m-i'}$ for $i'>i$.
        Let $i,i'$ be such that $m-i$ and $m-i'$ are both consecutive change points, i.e, 
        $m-i, m-i'\in N$ and $\forall \ell\in N. \ell<m-i' \vee \ell>m-i$. 
        Since the edge availability of the graph does not change between time $m-i'$ and $m-i$, we have $W_{m-i-1}=W_{m-i}$ implies $W_{m-i'}=W_{m-i}$.
        Therefore, we can accelerate the $\Pre$ computation and directly move to the time step $m-i'$, i.e, the $i'$th iteration in the computation. 
        This case is illustrated at time $n'=m-i'$ in \cref{fig:palg}.
        
	With this change, our algorithm runs the $\Pre$ computation at most $|V|+|N|$, as each $\Pre$ computation either corresponds to a step a time in $N$ when the graph changes, 
	or a step in which the winning region grows such as at time $n$ in \cref{fig:palg}.
	Since each $\Pre$ computation can be done in polynomial time, 
	we get a PTIME algorithm in this case, shown in \cref{alg:ptime}.

	 The case for improving temporal reachability games can be solved similarly. Instead of computing the winning region for \PONE in $G_m$, we start with computing the winning region $W_m^2$ for \PTWO in $G_m$ and switch the roles of \PONE and \PTWO, i.e, \PTWO has the punctual reachability objective with target set $W_m^2$ and target time $m$, which can be solved as above. This gives us an algorithm to compute the winning region for \PTWO and by determinacy of reachability games on infinite graphs, we can compute the winning region for \PONE at time $0$ as well.
	 \qed
\end{proof}

\begin{remark}
	\cref{alg:ptime} also works for parity objectives by changing step $1$, where $\text{Solve}(G_m)$ would amount to solving the parity game on the static graph $G_m$. This can be done in quasi-polynomial time and therefore gives a quasi-polynomial time algorithm to solve declining  (improving) temporal parity games and in particular, gives membership in the complexity class $\NP \cap \coNP$. 
\end{remark}
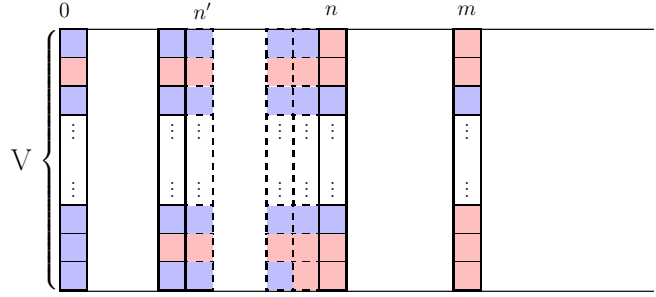
\begin{figure}[t]
	\setlength{\arrayrulewidth}{.1em}
	\begin{center}
		\resizebox{0.8\textwidth}{!}{
			\begin{tikzpicture}[-{Stealth[length=3mm, width=2mm]},node distance={30mm}, thick, main/.style = {draw, circle}]
				
				\node[main, minimum size = 3cm, draw=none](vertex) at (0.25,0){\fontsize{20}{20}\selectfont \states};
				\node[main, minimum size = 3cm, draw=none](roundminusone) at (9,0){\fontsize{25}{25}\selectfont };
				\node[main, draw=none](kzero) at (1.25,3.5){\Large{$0$}};
				\node[main, draw=none](kone) at (4.45,3.5){\Large{$n'$}};
				\node[main, draw=none](ktwo) at (7.45,3.5){\Large{$n$}};
				\node[main, draw=none](kthree) at (10.59,3.5){\Large{$m$}};
				
				\node[main,draw=none](table0) at (1.25,0){
					\Large
					$\left\{
					\begin{tabular}{|c|}
						\hline
						\cellcolor{blue!25}\rule{0.5cm}{0pt}\\
						\hline
						\cellcolor{red!25}\\
						\hline
						\cellcolor{blue!25}\\
						\hline
						$\vdots$\\
						\\
						$\vdots$\\
						\hline
						\cellcolor{blue!25}\\
						\hline
						\cellcolor{blue!25}\\
						\hline
						\cellcolor{blue!25}\\
						\hline
					\end{tabular}
					\right.$
				};
				\node[main,draw=none](table1) at (3.75,0){
					\Large
					\begin{tabular}{|c|}
						\hline
						\cellcolor{blue!25}\rule{0.5cm}{0pt}\\
						\hline
						\cellcolor{red!25}\\
						\hline
						\cellcolor{blue!25}\\
						\hline
						$\vdots$\\
						\\
						$\vdots$\\
						\hline
						\cellcolor{blue!25}\\
						\hline
						\cellcolor{red!25}\\
						\hline
						\cellcolor{blue!25}\\
						\hline
					\end{tabular}
				};
				\node[main,draw=none](table1) at (4.37,0){
					\Large
					\begin{tabular}{|c:}
						\hdashline
						\cellcolor{blue!25}\rule{0.5cm}{0pt}\\
						\hdashline
						\cellcolor{red!25}\\
						\hdashline
						\cellcolor{blue!25}\\
						\hdashline
						$\vdots$\\
						\\
						$\vdots$\\
						\hdashline
						\cellcolor{blue!25}\\
						\hdashline
						\cellcolor{red!25}\\
						\hdashline
						\cellcolor{blue!25}\\
						\hline
					\end{tabular}
				};
				\node[main,draw=none](table2) at (6.25,0){
					\Large
					\begin{tabular}{:c:}
						\hdashline
						\cellcolor{blue!25}\rule{0.5cm}{0pt}\\
						\hdashline
						\cellcolor{red!25}\\
						\hdashline
						\cellcolor{blue!25}\\
						\hdashline
						$\vdots$\\
						\\
						$\vdots$\\
						\hdashline
						\cellcolor{blue!25}\\
						\hdashline
						\cellcolor{red!25}\\
						\hdashline
						\cellcolor{blue!25}\\
						\hdashline
					\end{tabular}
				};
				\node[main,draw=none](table2) at (6.85,0){
					\Large
					\begin{tabular}{:c}
						\hdashline
						\cellcolor{blue!25}\rule{0.5cm}{0pt}\\
						\hdashline
						\cellcolor{red!25}\\
						\hdashline
						\cellcolor{blue!25}\\
						\hdashline
						$\vdots$\\
						\\
						$\vdots$\\
						\hdashline
						\cellcolor{blue!25}\\
						\hdashline
						\cellcolor{red!25}\\
						\hdashline
						\cellcolor{red!25}\\
						\hdashline
					\end{tabular}
				};
				\node[main,draw=none](table2) at (7.47,0){
					\Large
					\begin{tabular}{|c|}
						\hline
						\cellcolor{red!25}\rule{0.5cm}{0pt}\\
						\hline
						\cellcolor{red!25}\\
						\hline
						\cellcolor{blue!25}\\
						\hline
						$\vdots$\\
						\\
						$\vdots$\\
						\hline
						\cellcolor{blue!25}\\
						\hline
						\cellcolor{red!25}\\
						\hline
						\cellcolor{red!25}\\
						\hline
					\end{tabular}
				};
				\node[main,draw=none](table3) at (10.59,0){
					\Large
					\begin{tabular}{|c|}
						\hline
						\cellcolor{red!25}\rule{0.5cm}{0pt}\\
						\hline
						\cellcolor{red!25}\\
						\hline
						\cellcolor{blue!25}\\
						\hline
						$\vdots$\\
						\\
						$\vdots$\\
						\hline
						\cellcolor{red!25}\\
						\hline
						\cellcolor{red!25}\\
						\hline
						\cellcolor{red!25}\\
						\hline
					\end{tabular}
				};
				
				\path[every loop/.append style=-{Stealth[length=3mm, width=2mm]}]
				(1.25,3.06) edge [-] node[]{} (15,3.06)
				(1.25,-3.07) edge [-] node[]{} (15,-3.07)
				;
			\end{tikzpicture}
		}
	\end{center}
	\caption{Illustration of \cref{alg:ptime}. The blue vertices at time $i$  denote the winning region $W_i$ for \PONE. The times $n,n'\in N$ and $\Pre$ computation at change point $n$ increases the winning region but is stable at time $n'$.  }
	\label{fig:palg}
\end{figure}

Since the declining (improving) restriction on games on temporal graphs allow
for improved algorithms,
a natural question is to try to lift this approach to a larger class of games on temporal graphs. 
Note that the above restrictions are a special case of eventually periodic temporal graphs with a prefix of time $m$ followed by a periodic graph with period $1$. 
Now, we consider temporal graphs of period $\period>1$
such that the game arena is declining (improving) within each period.
Formally, a game on a temporal graph $G$ is \emph{periodically declining} (improving) if 
there exists a period $\period$ such that for all $k\in \N$, $k\in\edges(u,v)$ if and only if $k+\period\in \edges(u,v)$; and
the game on the finite temporal graph resulting from $G$ by making the graph constant from time $\period$ onwards, is declining (improving).
We prove that this case is \PSPACE-hard, even with reachability objectives.

\begin{theorem}
    \label{thm:monotone-declining}
	Solving  periodically declining (improving) temporal reachability games  is \PSPACE-complete.
\end{theorem}

\begin{proof}
	The upper bound follows from the general case of parity games on periodic temporal graphs in \cref{thm:parity}. 
	The lower bound is by reduction from punctual reachability games.
	See \cref{FigureTwoB}.
	Given a (static) graph $G$ with target state $v$ and target time $T$, we obtain a periodically declining game $G'$ with period $\period=T+1$, vertices $V\cup \{w,\bot,\top\}$, new target $\top$, such that $V'_1=V_1\cup \{w,\bot,\top\}$ and $V'_2=V_2$. 
	We assume without loss of generality that the original target $v$ is a \PONE vertex, i.e, $v\in V_1$.
	
	We describe the edge availability in $G'$ up to 
	the period $\period= T+1$. 
	For all edges $(s,t)$ of the original graph $G$, such that $s\in V_1$, the edge $s\step{x}t$ is available if and only if $x<T$. Moreover for 
	all $s\in V_1\setminus \{v\}$, there is a new edge $s\step{x}\bot$ available at all times $x\le T$. 	 	
	For all $s\in V_2$, there is an edge $s\step{x}t$ is available at all times (until end of period) and $s\step{x}\bot$ is available after time $x\ge T$.
	These edges ensure that if a play in the original punctual reachability game ends in a vertex of the game other than $v$ at time $T$, 
	then \PTWO can force the play to reach the sink state $\bot$ and win.
	
	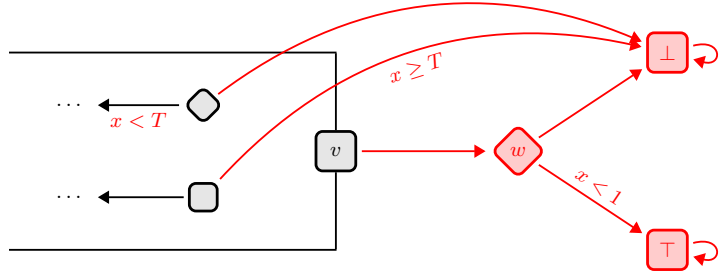
\begin{figure}[t]
		\centering
		\resizebox{0.8\textwidth}{!}{
						\begin{tikzpicture}[]
				
				\node[ostate] (v) at (10, 0){$v$};
				\node[estate, color=red, fill=red!20] (w) at (12.75, 0){$w$};
				\node[estate, scale=0.7] (player0) at (8, 0.7){};
				\node[ostate, scale=0.7] (player1) at (8,-0.7){};
				\node[draw=none] (player0dest) at (6, 0.7){$\dots$};
				\node[draw=none] (player1dest) at (6, -0.7){$\dots$};
				\node[ostate, color=red, fill=red!20] (top) at (15, -1.5){$\top$};
				\node[ostate, color=red, fill=red!20] (bot) at (15, 1.5){$\bot$};
				
				\path[every loop/.append style=-{>}]
				(5, 1.5) edge [-] node [above]{} (10.08, 1.5)
				(5, -1.5) edge [-] node [above]{} (10.08, -1.5)
				(10, 1.56) edge [-] node [above]{} (10, 0.25)
				(10, -1.56) edge [-] node [above]{} (10, -0.25)
				(player0) edge [] node [below, color=red]{$x<T$} (player0dest)
				(v) edge [color=red] node [above]{} (w)
				(player1) edge [] node [above]{} (player1dest)
				(w) edge [color=red] node [above, color=red, sloped]{$x<1$} (top)
				(w) edge [color=red] node [above]{} (bot)
				(bot) edge [loop right, color=red] node [above]{} (bot)
				(top) edge [loop right, color=red] node [above]{} (top)
				(player0) edge [bend left=30, color=red] node [above]{} (bot)
				(player1) edge [bend left, color=red] node [below, color=red, sloped]{$x\geq T$} (bot)
				;
			\end{tikzpicture}
		}
		\caption{Reduction from a punctual reachability game to a
			reachability game on a temporal graphs that is
			periodic and declining, see \cref{thm:monotone-declining}.
			Parts added are shown in red.
		}
		\label{FigureTwoB}
	\end{figure}
	
	From the original target $v$, there is an edge to the new state 
	$w$ at all times.
	From the state $w$, there are edges $w\step{}\bot$ at all times and $w\step{x}\top$ if $x= 0$. 
	If the state $w$ is reached at time $k$ such that $1<k<T+1$, then the play is forced to go to $\bot$. 
	The only winning strategy for \PONE is to reach $v$ at time $T$, $w$ at time $T+1$ at which the time is reset due to periodicity. The edge $w\step{T+1}\top$ is now available for \PONE and they can reach the new target $\top$.
	
	The lower bound for the case of periodically increasing temporal reachability games follows by the same construction and using the duality between improving and declining games on temporal graphs. 
	Given a punctual reachability game $G$ with vertices $V=V_1\uplus V_2$ with target set $F$, we obtain the dual punctual reachability game $\hat{G}$ 
	with same target time by first switch the ownership of vertices, i.e, $\hat{V}_i=V_{3-i}$, $i\in \{1,2\}$ and make the new target as $V\setminus F$. It is easy to see that \PONE wins $G$ if and only if \PTWO wins $\hat{G}$.
	
	Applying the same construction as shown in \cref{FigureTwoB} to $\hat{G}$, we obtain 
	a periodically declining temporal reachability game $\hat{G'}$, preserving the winner. 
	Now switching the ownership of vertices in $\hat{G'}$ yields a periodically improving temporal reachability game $G'$ which is winning for \PONE if and only if \PONE wins $G$. 
	\qed
\end{proof}

\section{Conclusion}
\label{sec:conclusion}
In this work we showed that parity games on ultimately periodic temporal graphs
are solvable in polynomial space.
The lower bound already holds for the very special case of punctual reachability games,
and the \PSPACE\ upper bound, which improves on the na\"ive exponential-space algorithm
on the unfolded graph, is achieved by proving the existence of small, \PSPACE-verifiable certificates.

We stress again that all constructions are effective no matter how the temporal graphs are defined,
as long as checking edge availability for binary encoded times is no obstacle.
In the paper we use edge constraints given in the existential fragment of Presburger arithmetic but alternate representations,
for example using compressed binary strings of length $\hor(G)$ given as Straight-Line Programs~\cite[Section~3]{BS84}
would equally work. Checking existence of edge at time $i$ would correspond to querying whether the $i^{th}$ bit is $1$ or not which is \P-complete \cite[Theorem~1]{LL06}. 

The games considered here are somewhat orthogonal to parity games played on the configuration graphs of timed automata,
where time is continuous, and constraints are \emph{quantifier-free} formulae involving possibly more than one variable (clocks).
Solving parity games on timed automata with two clocks 
is complete for \EXP\ but is in \P\ if there is at most one one clock \cite{AT09}~\cite[Contribution~3(d)]{HIM2013}.
Games on temporal graphs with quantifier-free constraints corresponds to a subclass of timed automata 
games with two-clocks, with intermediate complexity of \PSPACE.
This is because translating a temporal graph game to a timed automata game requires two clocks: one to hold the global time used to
check the edge predicate and one to ensure that time progresses one unit per step.

An interesting continuation of the work presented here would be to consider mean-payoff games
\cite{EM79}  
played on temporal graphs, possibly with dynamic step-rewards depending on the time.
If rewards are constant but the edge availability is dynamic, then our arguments for improved algorithms on declining/improving graphs would easily transfer.
However, the \PSPACE\ upper bound using summaries seems trickier, particularly checking realisability of suitable certificates.

\subsubsection{Acknowledgements} 
This work was supported by the 
Engineering and Physical Sciences Research Council (EPSRC), grant
\verb|EP/V025848/1|.
We thank Viktor Zamaraev and Sven Schewe
for fruitful discussions and constructive feedback.
%
%
%
\bibliographystyle{splncs04}
\bibliography{bib/journals,bib/conferences,bib/references}

\begin{thebibliography}{10}
\providecommand{\url}[1]{\texttt{#1}}
\providecommand{\urlprefix}{URL }
\providecommand{\doi}[1]{https://doi.org/#1}

\bibitem{GAMES2002}
Automata Logics, and Infinite Games: A Guide to Current Research.
  Springer-Verlag (2002)

\bibitem{AMGZ20}
Akrida, E.C., Mertzios, G.B., Spirakis, P.G., Zamaraev, V.: Temporal vertex
  cover with a sliding time window. Journal of Computer and System Sciences
  \textbf{107},  108--123 (2020).
  \doi{https://doi.org/10.1016/j.jcss.2019.08.002}

\bibitem{AD1994}
Alur, R., Dill, D.L.: A theory of timed automata. Theor. Comput. Sci.
  \textbf{126}(2),  183 -- 235 (1994). \doi{10.1016/0304-3975(94)90010-8}

\bibitem{AGG2023}
Avni, G., Ghorpade, P., Guha, S.: {A Game of Pawns}. In: International
  Conference on Concurrency Theory. Leibniz International Proceedings in
  Informatics (LIPIcs), vol.~279, pp. 16:1--16:17. Schloss Dagstuhl --
  Leibniz-Zentrum f{\"u}r Informatik (2023).
  \doi{10.4230/LIPIcs.CONCUR.2023.16}

\bibitem{BS84}
Babai, L., Szemeredi, E.: On the complexity of matrix group problems i. In:
  Annual Symposium on Foundations of Computer Science. pp. 229--240 (1984).
  \doi{10.1109/SFCS.1984.715919}

\bibitem{CJKLS2017}
Calude, C.S., Jain, S., Khoussainov, B., Li, W., Stephan, F.: Deciding parity
  games in quasipolynomial time. In: Symposium on Theory of Computing. pp.
  252--263 (2017). \doi{10.1145/3055399.3055409}

\bibitem{CKS1981}
Chandra, A.K., Kozen, D.C., Stockmeyer, L.J.: Alternation. Journal of the ACM
  (JACM)  \textbf{28}(1),  114--133 (1981)

\bibitem{CHP2011}
Chatterjee, K., Henzinger, T.A., Prabhu, V.S.: {Timed Parity Games: Complexity
  and Robustness}. Logical Methods in Computer Science  \textbf{{Volume 7,
  Issue 4}} (Dec 2011). \doi{10.2168/LMCS-7(4:8)2011}

\bibitem{CF2019}
Colcombet, T., Fijalkow, N.: {Universal Graphs and Good for Games Automata: New
  Tools for Infinite Duration Games}. In: International Conference on
  Foundations of Software Science and Computational Structures. LNCS, vol.
  11425, pp. 1--26. Springer (2019 DOI:
  \href{https://doiorg/101007/978-3-030-17127-8\_1}{101007/978-3-030-17127-8\_1}).
  \doi{10.1007/978-3-030-17127-8\_1}

\bibitem{DFS2023}
De~Carufel, J.L., Flocchini, P., Santoro, N., Simard, F.: Cops {\&} robber
  on periodic temporal graphs: Characterization and improved bounds. In:
  Structural Information and Communication Complexity. pp. 386--405. Springer
  Nature Switzerland (2023)

\bibitem{EM79}
Ehrenfeucht, A., Mycielski, J.: Positional strategies for mean payoff games.
  International Journal of Game Theory  \textbf{8}(2),  109--113 (Jun 1979).
  \doi{10.1007/BF01768705}

\bibitem{EHK21}
Erlebach, T., Hoffmann, M., Kammer, F.: On temporal graph exploration. Journal
  of Computer and System Sciences  \textbf{119},  1--18 (2021).
  \doi{https://doi.org/10.1016/j.jcss.2021.01.005}

\bibitem{GAMES2023}
Fijalkow, N., Bertrand, N., Bouyer-Decitre, P., Brenguier, R., Carayol, A.,
  Fearnley, J., Gimbert, H., Horn, F., Ibsen-Jensen, R., Markey, N., Monmege,
  B., Novotný, P., Randour, M., Sankur, O., Schmitz, S., Serre, O., Skomra,
  M.: Games on graphs (2023)

\bibitem{FMS2009}
Flocchini, P., Mans, B., Santoro, N.: Exploration of periodically varying
  graphs. In: Algorithms and Computation. pp. 534--543. Springer Berlin
  Heidelberg (2009)

\bibitem{H2018}
Haase, C.: A survival guide to presburger arithmetic. SIGLOG News
  \textbf{5}(3),  67--82 (2018). \doi{10.1145/3242953.3242964}

\bibitem{HIM2013}
Hansen, T.D., Ibsen-Jensen, R., Miltersen, P.B.: A faster algorithm for solving
  one-clock priced timed games (2013)

\bibitem{HJ19}
Holme, P., Saramäki, J.: Temporal Network Theory (01 2019).
  \doi{10.1007/978-3-030-23495-9}

\bibitem{H1995}
Holzer, M.: On emptiness and counting for alternating finite automata. In:
  International Conference on Developments in Language Theory. pp. 88--97
  (1995)

\bibitem{JS2007}
Jan\'{c}ar, P., Sawa, Z.: A note on emptiness for alternating finite automata
  with a one-letter alphabet. Inf. Process. Lett.  \textbf{104}(5),  164--167
  (2007). \doi{https://doi.org/10.1016/j.ipl.2007.06.006}

\bibitem{JR1991}
Jiang, T., Ravikumar, B.: A note on the space complexity of some decision
  problems for finite automata. Inf. Process. Lett.  \textbf{40}(1),  25--31
  (1991). \doi{https://doi.org/10.1016/S0020-0190(05)80006-7}

\bibitem{JT2007}
Jurdzi{\'{n}}ski, M., Trivedi, A.: Reachability-time games on timed automata.
  In: International Colloquium on Automata, Languages and Programming. pp.
  838--849. Springer Berlin Heidelberg (2007)

\bibitem{J1998}
Jurdziński, M.: Deciding the winner in parity games is in up $\cap$ co-up.
  Inf. Process. Lett.  \textbf{68}(3),  119--124 (1998).
  \doi{https://doi.org/10.1016/S0020-0190(98)00150-1}

\bibitem{JL2017}
Jurdziński, M., Lazić, R.: {Succinct Progress Measures for Solving Parity
  Games}. In: Annual IEEE Symposium on Logic in Computer Science. pp.~1--9.
  IEEE Computer Society (2017 DOI:
  \href{https://doiorg/101109/LICS20178005092}{101109/LICS20178005092}).
  \doi{10.1109/LICS.2017.8005092}

\bibitem{KLO2010}
Kuhn, F., Lynch, N., Oshman, R.: Distributed computation in dynamic networks.
  In: Symposium on Theory of Computing. p. 513–522. STOC '10, Association for
  Computing Machinery (2010). \doi{10.1145/1806689.1806760}

\bibitem{LPSD2022}
Lehtinen, K., Parys, P., Schewe, S., Wojtczak, D.: {A Recursive Approach to
  Solving Parity Games in Quasipolynomial Time}. Logical Methods in Computer
  Science  \textbf{18}(1),  8:1--18 (2022)

\bibitem{LB2020}
Lehtinen, K., Boker, U.: {Register Games}. Logical Methods in Computer Science
  \textbf{16}(2),  6:1--6:25 (2020 DOI:
  \href{https://doiorg/1023638/LMCS-16(2:6)2020}{1023638/LMCS-16(2:6)2020}).
  \doi{10.23638/LMCS-16(2:6)2020}

\bibitem{LL06}
Lifshits, Y., Lohrey, M.: Querying and embedding compressed texts. In:
  International Symposium on Mathematical Foundations of Computer Science. pp.
  681--692. Springer Berlin Heidelberg (2006)

\bibitem{MPS1995}
Maler, O., Pnueli, A., Sifakis, J.: On the synthesis of discrete controllers
  for timed systems. In: International Symposium on Theoretical Aspects of
  Computer Science. pp. 229--242. Springer Berlin Heidelberg (1995)

\bibitem{MMN23}
Mertzios, G.B., Molter, H., Niedermeier, R., Zamaraev, V., Zschoche, P.:
  Computing maximum matchings in temporal graphs. Journal of Computer and
  System Sciences  \textbf{137},  1--19 (2023).
  \doi{https://doi.org/10.1016/j.jcss.2023.04.005}

\bibitem{MMZ21}
Mertzios, G.B., Molter, H., Zamaraev, V.: Sliding window temporal graph
  coloring. Journal of Computer and System Sciences  \textbf{120},  97--115
  (2021). \doi{https://doi.org/10.1016/j.jcss.2021.03.005}

\bibitem{M2015}
Michail, O.: An Introduction to Temporal Graphs: An Algorithmic Perspective,
  pp. 308--343. Springer International Publishing (2015).
  \doi{10.1007/978-3-319-24024-4_18}

\bibitem{MCS2014}
Michail, O., Chatzigiannakis, I., Spirakis, P.G.: Causality, influence, and
  computation in possibly disconnected synchronous dynamic networks. Journal of
  Parallel and Distributed Computing  \textbf{74}(1),  2016--2026 (2014)

\bibitem{MS14}
Michail, O., Spirakis, P.G.: Traveling salesman problems in temporal graphs.
  In: International Symposium on Mathematical Foundations of Computer Science.
  pp. 553--564. Springer Berlin Heidelberg (2014)

\bibitem{PR1989}
Pnueli, A., Rosner, R.: On the synthesis of a reactive module. In: Annual
  Symposium on Principles of Programming Languages. p. 179–190. POPL '89,
  Association for Computing Machinery (1989). \doi{10.1145/75277.75293}

\bibitem{P1977}
Pnueli, A.: The temporal logic of programs. In: Annual Symposium on Foundations
  of Computer Science. p. 46–57. SFCS '77, IEEE Computer Society (1977).
  \doi{10.1109/SFCS.1977.32}

\bibitem{ravi1994}
Ravi, R.: Rapid rumor ramification: Approximating the minimum broadcast time.
  In: Proceedings 35th Annual Symposium on Foundations of Computer Science. pp.
  202--213 (1994)

\bibitem{S1984}
Scarpellini, B.: Complexity of subcases of presburger arithmetic. Transactions
  of the American Mathematical Society  \textbf{284},  203–218 (1984).
  \doi{10.1090/s0002-9947-1984-0742421-9}

\bibitem{AT09}
Trivedi, A.: Competitive optimisation on timed automata. Ph.D. thesis,
  University of Warwick (April 2009)

\end{thebibliography}
\end{document}